\documentclass[final,5p,times,twocolumn]{elsarticle}

\usepackage[T1]{fontenc}
\usepackage{amsmath}
\usepackage{amssymb}
\usepackage{amsthm}
\usepackage[hidelinks]{hyperref}

\biboptions{sort&compress}

\newtheorem{theorem}{Theorem}
\newtheorem{lemma}[theorem]{Lemma}
\newtheorem{proposition}[theorem]{Proposition}
\newtheorem{corollary}[theorem]{Corollary}

\newcommand{\sgn}{\operatorname{sgn}}
\newcommand{\Q}{\mathcal{Q}}
\newcommand{\R}{\mathbb{R}}
\newcommand{\Z}{\mathbb{Z}}
\newcommand{\Sn}{S_n}

\newcounter{iplbibitem}
\pretocmd{\bibitem}{%
  \stepcounter{iplbibitem}%
  \ifnumcomp{\value{iplbibitem}}{=}{11}{\newpage}{}%
}{}{}

\journal{Information Processing Letters}

\begin{document}

\begin{frontmatter}

\title{Polynomial-Time Singular Witnesses for Non-SNS Sign Patterns}

\author[aff1,aff2]{Tao Jiang}
\ead{jiangt@ios.ac.cn}

\author[aff1,aff2]{Minbo Gao}
\ead{gaomb@ios.ac.cn}

\author[aff1,aff2]{Shaowei Cai\texorpdfstring{\corref{cor1}}{}}
\ead{caisw@ios.ac.cn}
\cortext[cor1]{Corresponding author}

\affiliation[aff1]{
  organization={Key Laboratory of System Software (Chinese Academy of
  Sciences) and State Key Laboratory of Computer Science, Institute of
  Software, Chinese Academy of Sciences},
  addressline={No. 4 South Fourth Street, Zhongguancun, Haidian District},
  city={Beijing},
  postcode={100190},
  country={China}}

\affiliation[aff2]{
  organization={School of Computer Science and Technology, University of
  Chinese Academy of Sciences},
  addressline={No. 19A Yuquan Road, Shijingshan District},
  city={Beijing},
  postcode={100049},
  country={China}}

\begin{abstract}
Sign-nonsingularity asks whether every real matrix with prescribed entry signs
is nonsingular. Polynomial-time algorithms recognize square sign-nonsingular
patterns through their connection with even directed cycles, but recognition
does not itself produce an exact numerical witness in the negative case. We
give a deterministic polynomial-time algorithm that, for any square sign
pattern $A$, either reports that $A$ is sign-nonsingular or outputs
$B\in\Z^{n\times n}$ and $z\in\Z^n\setminus\{0\}$ such that
$\sgn(B)=A$ and $Bz=0$. After normalizing a perfect matching, an even directed
cycle yields two determinant terms of opposite signs. Making either term
dominant produces endpoint realizations with opposite determinant signs.
Changing their magnitudes one coordinate at a time exposes an affine
sign-changing step, whose zero is rational; clearing its denominator gives
the integer witness. Entries of $B$ have $O(n^2\log n)$ bits, and entries of
$z$ have $O(n^3\log n)$ bits. The result settles Conjecture 14.12.4 in the
\emph{Handbook of Satisfiability}.
\end{abstract}

\begin{keyword}
sign-nonsingular matrix \sep sign pattern \sep singular realization \sep
even directed cycle \sep polynomial-time algorithm
\end{keyword}

\end{frontmatter}

\section{Introduction}
\label{sec:introduction}

For a sign pattern $A\in\{-1,0,+1\}^{n\times n}$, its qualitative class is
\[
  \Q(A)=\{B\in\R^{n\times n}:\sgn(B_{ij})=A_{ij}
  \text{ for all }i,j\}.
\]
The pattern is \emph{sign-nonsingular} (SNS) if every matrix in $\Q(A)$ is
nonsingular. Sign patterns arose in qualitative economic reasoning and the
study of sign-solvable linear systems
\cite{BassettMaybeeQuirk1968,KleeLadnerManber1984,BrualdiShader1995}.
For square patterns, SNS is also tied to P\'olya's permanent problem, Pfaffian
orientations, and even directed cycles
\cite{Polya1913,Kasteleyn1967,Little1975,VaziraniYannakakis1989}.

Structural work on these equivalent problems culminated in polynomial-time
recognition algorithms
\cite{SeymourThomassen1987,Thomassen1992,RobertsonSeymourThomas1999,
McCuaig2004}. These results decide whether a pattern is SNS and describe
obstructions when it is not. The functional question considered here asks for
a different output: from a non-SNS pattern, construct an exactly represented
singular member of its qualitative class.

The 2021 \emph{Handbook of Satisfiability} posed this functional witness
problem as Conjecture 14.12.4 \cite{KleineBuningKullmann2021}. To the best of
our knowledge, it has remained open. We resolve the conjecture for square sign
patterns by giving an explicit polynomial-time construction. Its motivation
comes from the connection between square SNS patterns, complement-invariant
satisfiability, deficiency, and balanced autarkies
\cite{FleischnerKullmannSzeider2002,Kullmann2003,Kullmann2007}. The precise
gap between the known recognition and equivalence results and the functional
witness problem is discussed in Section~\ref{sec:relation}.

Our construction first extracts two opposite-sign determinant terms using the
known even-directed-cycle machinery. Assigning a large magnitude to the
entries of either term gives two integer realizations whose determinants have
opposite signs. A direct segment between them is unsuitable because its
determinant has degree up to $n$. Instead, we change one entry magnitude at a
time. On each such step the determinant is affine, so a sign-changing step has
an exact rational zero. Clearing its single denominator preserves the sign
pattern and gives an integer matrix. Exact elimination then supplies a short
integer null vector.

\begin{theorem}
\label{thm:main}
There is a deterministic polynomial-time algorithm that, given
$A\in\{-1,0,+1\}^{n\times n}$, either reports that $A$ is SNS or outputs
\[
  B\in\Z^{n\times n},\qquad z\in\Z^n\setminus\{0\},
  \qquad \sgn(B)=A,\qquad Bz=0.
\]
In the latter case, entries of $B$ have $O(n^2\log n)$ bits and entries of
$z$ have $O(n^3\log n)$ bits.
\end{theorem}

The contribution is therefore certificate extraction rather than a new
structure theorem for even directed cycles: recognition is converted into a
polynomial-bit algebraic witness for the negative SNS case.

\section{Relation to previous work}
\label{sec:relation}

The determinant-term criterion underlying SNS is classical
\cite{KleeLadnerManber1984,BrualdiShader1995}. Its direct use is not an
efficient algorithm, since a matrix may have exponentially many nonzero
terms. The graph-theoretic line from convertible matrices and Pfaffian
orientations to even directed cycles replaces this enumeration by structural
recognition
\cite{Little1975,VaziraniYannakakis1989,SeymourThomassen1987,
Thomassen1992}. Robertson, Seymour, and Thomas, and independently the later
development by McCuaig, supply the structural machinery from which
polynomial-time recognition follows
\cite{RobertsonSeymourThomas1999,McCuaig2004}. We use that machinery only as
the black box in Theorem~\ref{thm:even-cycle-black-box}; no new even-cycle
structure theorem is claimed.

McCuaig's Theorem 4 is especially close to the present result. Among other
statements, it equates non-SNS with the existence of opposite-sign determinant
terms, a real singular realization, and a row-balanced signed column scaling
\cite{McCuaig2004}. Its large-term argument gives realizations with opposite
determinant signs, after which continuity gives a real zero on the segment
joining them. A determinant on that segment can have degree $n$, however, so
continuity alone neither makes the zero rational nor bounds an exact encoding.
Conversely, the passage from a singular realization to a row-balanced scaling
starts with that realization and a kernel vector already in hand. These
equivalences therefore do not provide the functional map required by the
Handbook conjecture: given $A$, output an integer singular member of $\Q(A)$
with polynomially many bits.

Our new ingredient begins after an opposite term has been extracted. Replacing
the single high-degree segment by a coordinate path makes every edge
multiaffine in only one active variable, hence affine. The sign-changing edge
has a rational root obtained from two integer determinant values. The explicit
size analysis then shows that clearing its denominator is polynomially
bounded, and exact elimination adds the independently verifiable null vector.
Thus any polynomial-time procedure that outputs an opposite term can be used
as the combinatorial front end of our algebraic certificate construction.

\section{Determinant terms and cycle extraction}
\label{sec:cycles}

\subsection{Terms and normalization}

A nonzero determinant term of $A$ is indexed by a permutation $\pi\in\Sn$
with $A_{i,\pi(i)}\ne0$ for every $i$. Its sign is
\[
  c_\pi=\sgn(\pi)\prod_{i=1}^n A_{i,\pi(i)}\in\{-1,+1\}.
\]

\begin{proposition}[Term criterion]
\label{prop:term-criterion}
A square sign pattern is SNS if and only if it has at least one nonzero
determinant term and all its nonzero terms have the same sign.
\end{proposition}

\begin{proof}
With no nonzero term, the determinant polynomial is identically zero. If all
nonzero terms have the same sign, every realization has a nonempty determinant
sum with that sign. Conversely, two opposite-sign terms yield a singular
realization by Lemmas~\ref{lem:domination} and \ref{lem:interpolation}.
\end{proof}

Let $G_A$ be the bipartite support graph of $A$. A nonzero determinant term is
exactly a perfect matching of $G_A$. Thus, if $G_A$ has no perfect matching,
$A$ itself is an integer singular witness. Otherwise choose a perfect matching
$\mu$, permute columns so that $\mu$ is the diagonal, and multiply each row by
the sign of its diagonal entry. The resulting pattern
\[
  C=DAP
\]
has $C_{ii}=+1$ for all $i$, where $D$ is a diagonal sign matrix and $P$ is a
permutation matrix. A singular $X\in\Q(C)$ transforms back to
$DXP^{\mathsf T}\in\Q(A)$, so it suffices to treat the positive-diagonal case.

\subsection{The even-cycle black box}

We use the following known algorithmic result.

\begin{theorem}[Even-directed-cycle black box]
\label{thm:even-cycle-black-box}
Given a finite digraph $H$, one can decide in deterministic polynomial time
whether $H$ contains an even directed cycle. If one exists, such a cycle can
be output in deterministic polynomial time.
\end{theorem}

\begin{proof}
The decision result follows from the structural and algorithmic theory in
\cite{VaziraniYannakakis1989,RobertsonSeymourThomas1999,McCuaig2004}; see
also the later shortest-even-cycle algorithm in
\cite{BjorklundHusfeldtKaski2022}. Decision gives search by edge deletion:
delete an arc whenever the remaining graph still contains an even directed
cycle. At termination the remaining graph $H^*$ is edge-minimal with this
property. If $C$ is any even directed cycle in $H^*$, an arc outside $C$ could
still be deleted. Hence every nonisolated arc of $H^*$ lies on $C$, and $C$
can be read off by traversal. At most $|E(H)|+1$ decision calls are used.
\end{proof}

Assume now that $C_{ii}=+1$ for all $i$. Form a signed digraph $D(C)$ on
$\{1,\ldots,n\}$ with an arc $i\to j$ for every off-diagonal nonzero entry
$C_{ij}$. Replace each positive arc by a path of length one and each negative
arc by a path of length two, using a distinct subdivision vertex for every
negative arc. Write $\widetilde D(C)$ for the resulting unsigned digraph.

\begin{lemma}
\label{lem:even-cycle}
The pattern $C$ has a determinant term opposite in sign to the identity term
if and only if $\widetilde D(C)$ contains an even directed cycle. Such a cycle
yields an opposite-sign term in linear time.
\end{lemma}

\begin{proof}
Consider a directed cycle $\gamma$ of length $k$ in $D(C)$, containing $r$
negative arcs. The associated permutation cycle, with all other vertices
fixed, has sign relative to the identity term
\[
  (-1)^{k-1}\prod_{e\in\gamma}C_e=(-1)^{k-1+r}.
\]
After subdivision, $\gamma$ has length $k+r$. This length is even exactly
when the displayed relative sign is $-1$.

Each subdivision vertex has one incoming and one outgoing arc, namely the two
arcs replacing one negative arc. A directed cycle visiting that vertex must
therefore use the entire length-two path. Contracting subdivision vertices in
an even cycle of $\widetilde D(C)$ gives a cycle of $D(C)$, to which the parity
calculation applies.

Conversely, decompose the permutation of any nonidentity determinant term
into disjoint directed cycles in $D(C)$. The term's sign relative to the
identity is the product of the cycle-relative signs. If it is negative, at
least one constituent cycle has relative sign $-1$, and its subdivision is an
even directed cycle.
\end{proof}

Theorem~\ref{thm:even-cycle-black-box} and Lemma~\ref{lem:even-cycle} thus
either prove that every nonzero term has the identity sign or produce an
opposite-sign permutation $\pi_-$. In the latter case set
$\pi_+=\mathrm{id}$.

\section{Constructing an integer singular realization}
\label{sec:construction}

\begin{lemma}[Term domination]
\label{lem:domination}
Given nonzero determinant terms $\pi_+$ and $\pi_-$ of opposite signs, one can
construct integer matrices $B^+,B^-\in\Q(A)$ such that
\[
  \sgn\det(B^+)=c_{\pi_+},\qquad
  \sgn\det(B^-)=c_{\pi_-}.
\]
\end{lemma}

\begin{proof}
Set $M=n!+1$. For $B^+$, give magnitude $M$ to every entry
$(i,\pi_+(i))$ and magnitude $1$ to every other nonzero entry, preserving the
signs of $A$. The selected term has magnitude $M^n$. Every other term misses
at least one selected entry and has magnitude at most $M^{n-1}$. Since there
are at most $n!-1$ such terms,
\[
  (n!-1)M^{n-1}<M^n.
\]
The selected term therefore determines the determinant sign. Construct
$B^-$ in the same way. This is the standard large-term domination argument
used in \cite{McCuaig2004}, made quantitative here.
\end{proof}

\begin{lemma}[Coordinatewise interpolation]
\label{lem:interpolation}
Given integer matrices $B^+,B^-\in\Q(A)$ with opposite nonzero determinant
signs, one can compute an integer singular matrix $B\in\Q(A)$ using at most
$n^2+1$ exact determinant evaluations.
\end{lemma}

\begin{proof}
Enumerate the nonzero positions as $e_1,\ldots,e_m$, where $m\le n^2$, and
write $x^+$ and $x^-$ for the positive magnitude vectors of $B^+$ and $B^-$.
Define $x^0=x^+$ and obtain $x^k$ by replacing the first $k$ coordinates with
their values in $x^-$. Coordinates whose endpoint values agree may be
skipped. Let
\[
  D_k=\det(A\mathbin{\odot}x^k).
\]
The endpoint values have opposite signs. Hence either some $D_k=0$, already
giving an integer witness, or $D_{k-1}D_k<0$ for some consecutive pair.

On that edge only one positive magnitude $t$ changes, from $a$ to $b$. With
all other entries fixed the determinant is affine in $t$. Write its endpoint
values as $g(a)=u$ and $g(b)=v$, where $uv<0$. Its unique zero is
\[
  t_0=\frac{av-bu}{v-u}.
\]
It lies strictly between $a$ and $b$, so no nonzero entry changes sign. Write
$t_0=p/q$ in lowest terms with $q>0$. Replacing that magnitude by $t_0$ and
multiplying the whole matrix by $q$ gives an integer singular matrix in
$\Q(A)$.
\end{proof}

The coordinatewise path is the arithmetic step missing from a general
continuous interpolation: its sign-changing equation is linear, so the zero
is rational and exactly recoverable from two integer determinants.

\subsection{A worked example}
\label{sec:example}

Consider the positive-diagonal pattern
\[
 A=\begin{pmatrix}
 +&-&-\\
 -&+&-\\
 -&+&+
 \end{pmatrix}.
\]
The identity term is positive, whereas the transposition $(1\ 2)$ gives a
negative determinant term. For $n=3$, Lemma~\ref{lem:domination} uses
$M=3!+1=7$ and produces
\[
 B^+=\begin{pmatrix}
 7&-1&-1\\
 -1&7&-1\\
 -1&1&7
 \end{pmatrix},\qquad \det(B^+)=336,
\]
and
\[
 B^-=\begin{pmatrix}
 1&-7&-1\\
 -7&1&-1\\
 -1&1&7
 \end{pmatrix},\qquad \det(B^-)=-336.
\]
Use row-major order for the coordinate path. After changing the $(1,1)$
magnitude from $7$ to $1$, the determinant is $36$. On the next edge, only
the magnitude of the negative entry $(1,2)$ changes. The matrices on this edge
are
\[
 X(t)=\begin{pmatrix}
 1&-t&-1\\
 -1&7&-1\\
 -1&1&7
 \end{pmatrix},\qquad 1\le t\le7,
\]
and direct expansion gives $\det X(t)=44-8t$. Thus the determinant changes
from $36$ to $-12$ and vanishes at the nonintegral rational value
$t_0=11/2$. Clearing the denominator yields
\[
 B=2X(11/2)=\begin{pmatrix}
 2&-11&-2\\
 -2&14&-2\\
 -2&2&14
 \end{pmatrix}\in\Q(A).
\]
Indeed, $\det(B)=0$ and
\[
 B\begin{pmatrix}25\\4\\3\end{pmatrix}=0.
\]
The example exhibits all arithmetic features of the general construction:
opposite endpoint signs, a one-coordinate affine crossing, a rational root,
one global denominator clearing, and a directly checkable integer null
certificate.

\section{Bit complexity and proof of the main theorem}
\label{sec:complexity}

\begin{lemma}[Witness size]
\label{lem:bit-bound}
The integer matrix produced by Lemmas~\ref{lem:domination} and
\ref{lem:interpolation} has entries of bit-length $O(n^2\log n)$ and is
computable in polynomial bit complexity.
\end{lemma}

\begin{proof}
Since $M=n!+1$, we have $\log M=O(n\log n)$. Every matrix at a vertex of the
coordinate path has entries bounded by $M$. Thus each endpoint determinant on
an edge satisfies the Leibniz bound
\[
  |u|,|v|\le n!M^n
\]
and has $O(n^2\log n)$ bits. On a nonconstant edge,
$\{a,b\}=\{1,M\}$. Before reduction, the numerator and denominator of $t_0$
satisfy
\[
  |av-bu|\le 2M(n!M^n),\qquad
  |v-u|\le 2n!M^n.
\]
The reduced integers $p,q$ therefore have $O(n^2\log n)$ bits. After scaling
by $q$, every nonzero entry is bounded by $\max\{qM,|p|\}$ and has the same
asymptotic bit-length. Exact determinants can be evaluated in polynomial bit
complexity by fraction-free elimination.
\end{proof}

\begin{lemma}[Integer null certificate]
\label{lem:null-vector}
Let $B\in\Z^{n\times n}$ be singular and
$H=\max\{1,\max_{i,j}|B_{ij}|\}$. In polynomial time one can compute
$z\in\Z^n\setminus\{0\}$ with $Bz=0$ and entry bit-length
$O(n\log n+n\log H)$.
\end{lemma}

\begin{proof}
Let $r=\operatorname{rank}(B)<n$. For $r=0$, take a unit vector. Otherwise
choose a nonsingular submatrix $R=B[I,J]$ of order $r$ and an index
$k\notin J$. The columns indexed by $J$ form a basis of the column space.
Solve
\[
  Rx=-B[I,k]
\]
and set $z_J=x$, $z_k=1$, with all other free coordinates zero. The relation
holds on every row: its left-hand side is in the span of the basis columns and
vanishes on $I$, where restriction of that span is injective.

Cramer's rule and clearing denominators give an integer vector whose
coordinates are determinants of order $r$ with entries bounded by $H$.
Hadamard's inequality bounds them by $r^{r/2}H^r$, yielding the stated bit
bound. A rank profile and the vector are computable by exact elimination.
\end{proof}

\begin{proof}[Proof of Theorem~\ref{thm:main}]
Find a perfect matching in the support graph. If none exists, take $B=A$;
otherwise normalize the matching to a positive diagonal and construct
$\widetilde D(C)$. If the black box finds no even directed cycle,
Lemma~\ref{lem:even-cycle} and Proposition~\ref{prop:term-criterion} imply
that $A$ is SNS. If it finds one, Lemma~\ref{lem:even-cycle} converts it into
an opposite-sign term. Lemmas~\ref{lem:domination} and
\ref{lem:interpolation} then produce $B$, and Lemma~\ref{lem:bit-bound} gives
its $O(n^2\log n)$ entry bound. Undoing normalization preserves singularity,
the qualitative class, and bit-length.

Finally apply Lemma~\ref{lem:null-vector}. Here
$\log H=O(n^2\log n)$, so entries of $z$ have $O(n^3\log n)$ bits. Perfect
matching, graph construction, the polynomially many black-box calls, and all
exact arithmetic steps have polynomial bit complexity.
\end{proof}

The pair $(B,z)$ is independently checkable in polynomial time by testing the
entry signs, $z\ne0$, and the exact identity $Bz=0$; verification does not use
the even-cycle machinery.

\section{Consequences}
\label{sec:consequences}

The non-SNS branch of Theorem~\ref{thm:main} immediately resolves the matrix
witness problem that motivated the construction.

\begin{corollary}[Handbook Conjecture 14.12.4]
\label{cor:handbook}
Given a square non-SNS sign pattern
$A\in\{-1,0,+1\}^{n\times n}$, one can compute in polynomial time an integer
matrix $B\in\Q(A)$ with $\det(B)=0$ and polynomial-bit entries.
\end{corollary}

The Handbook also records consequences for complement-invariant clause-sets
and balanced autarkies. Let $d_r^*(F)$ denote maximal reduced deficiency and
$d^*(F)$ maximal deficiency, using its notation.

\begin{corollary}[Deficiency zero]
\label{cor:autarky}
For complement-invariant clause-sets $F$ with $d_r^*(F)=0$, a quasi-maximal
autarky can be computed and satisfiability can be decided in polynomial time.
Equivalently, for clause-sets with $d^*(F)=0$, a quasi-maximal balanced autarky
can be computed and NAESAT can be decided in polynomial time.
\end{corollary}

\begin{proof}
Corollary~\ref{cor:handbook} proves Conjecture 14.12.4. The claims are then
exactly the implications of Theorem 14.12.1 in
\cite{KleineBuningKullmann2021}, based on the matrix--SAT correspondence in
\cite{Kullmann2007}.
\end{proof}

In qualitative matrix analysis, Theorem~\ref{thm:main} supplies the natural
complement to an SNS guarantee: when the signs do not force nonsingularity, it
returns explicit polynomial-bit numerical data witnessing degeneracy. We make
no claim here for rectangular L-matrices or positive bounded deficiency; that
extension is Conjecture 14.12.5 in the Handbook.

\section{Conclusion}
\label{sec:conclusion}

Known even-cycle algorithms provide the combinatorial obstruction to SNS.
Term domination and coordinatewise affine interpolation turn that obstruction
into an exact integer singular realization, while exact elimination adds a
short null certificate. Extending certificate extraction from the square case
to rectangular patterns with a fixed excess of columns remains the principal
open direction.

\section*{Data availability}

No data were used for the research described in this article.

\section*{CRediT authorship contribution statement}

Tao Jiang: Conceptualization, Formal analysis, Investigation, Methodology,
Software, Validation, Writing -- original draft. Minbo Gao: Formal analysis,
Methodology, Validation, Writing -- review and editing. Shaowei Cai:
Conceptualization, Supervision, Validation, Writing -- review and editing.

\section*{Declaration of generative AI and AI-assisted technologies in the
manuscript preparation process}

During the preparation of this work, OpenAI Codex was used to assist with
manuscript organization, language editing, \LaTeX{} preparation, and
computational verification. The authors reviewed and edited the resulting
material, independently verified the mathematical arguments and cited
sources, and take full responsibility for the content of the article.


\bibliographystyle{elsarticle-num}
\bibliography{singular-witnesses}

\end{document}